\documentclass{llncs}[11] 

\usepackage{graphicx}
\usepackage{cite}
\usepackage{hyperref}
\usepackage{amsmath}
\usepackage{amssymb}
\usepackage{amsfonts}
\usepackage{comment}
\usepackage[nofillcomment,noend]{algorithm2e}
\usepackage{appendix}

\input{commands-tam.sty}

\newtheorem{observation}{Observation}
\newtheorem{notation}{Notation}

\newif\ifabstract
\newif\iffull

\abstractfalse

\ifabstract
	\fullfalse
\else
	\fulltrue
\fi

\newcounter{section-preserve}
\newcounter{theorem-preserve}
\newcounter{lemma-preserve}
\newcommand{\blank}[1]{}
\newtoks\magicAppendix
\magicAppendix={}
\newtoks\magictoks
\newif\iflater
\laterfalse
\ifabstract
\long\def\later#1{\magictoks={#1}%
  \edef\magictodo{\noexpand\magicAppendix={\the\magicAppendix \par
    \the\magictoks%
  }}
  \magictodo}
\long\def\both#1{\magictoks={#1}%
  \edef\magictodo{\noexpand\magicAppendix={\the\magicAppendix \par
    \noexpand\setcounter{theorem-preserve}{\noexpand\arabic{theorem}}%
    \noexpand\setcounter{theorem}{\arabic{theorem}}%
    \noexpand\setcounter{lemma-preserve}{\noexpand\arabic{lemma}}%
    \noexpand\setcounter{lemma}{\arabic{lemma}}%
    \noexpand\setcounter{section-preserve}{\noexpand\arabic{section}}%
    \noexpand\setcounter{section}{\arabic{section}}%
	\noexpand\let\noexpand\oldsection=\noexpand\thesection
	\noexpand\def\noexpand\thesection{\thesection}
	\noexpand\let\noexpand\oldlabel=\noexpand\label
	\noexpand\let\noexpand\label=\noexpand\blank
    \the\magictoks%
    \noexpand\setcounter{theorem}{\noexpand\arabic{theorem-preserve}}%
    \noexpand\setcounter{section}{\noexpand\arabic{section-preserve}}%
	\noexpand\let\noexpand\thesection=\noexpand\oldsection
	\noexpand\let\noexpand\label=\noexpand\oldlabel
  }}
  \magictodo
  \the\magictoks}
\else
\long\def\later#1{#1}
\long\def\both#1{#1}
\fi
\long\def\magicappendix{
	\latertrue%
	\the\magicAppendix%
}

\begin{document}

\title{\textbf{Scaled tree fractals do not strictly self-assemble}}

\author{%
Kimberly Barth%
    \thanks{Department of Computer Science, University of Wisconsin--Oshkosh, Oshkosh, WI 54901, USA,
    \protect\url{barthk63@uwosh.edu}.}
\and
David Furcy
    \thanks{Department of Computer Science, University of Wisconsin--Oshkosh, Oshkosh, WI 54901, USA,
    \protect\url{furcyd@uwosh.edu}.}
\and
Scott M. Summers
    \thanks{Department of Computer Science, University of Wisconsin--Oshkosh, Oshkosh, WI 54901, USA,
    \protect\url{summerss@uwosh.edu}.}
\and
Paul Totzke
    \thanks{Department of Computer Science, University of Wisconsin--Oshkosh, Oshkosh, WI 54901, USA,
    \protect\url{totzkp00@uwosh.edu}.}
}

\institute{}

\date{}
\maketitle

\sloppy

\begin{abstract}
In this paper, we show that any scaled-up version of any discrete self-similar \emph{tree} fractal does not strictly self-assemble, at any temperature, in Winfree's abstract Tile Assembly Model.
\end{abstract} 
\section{Introduction}
The stunning, often mysterious complexities of the natural world, from nanoscale crystalline structures to unthinkably massive galaxies, all arise from the same elemental process known as \emph{self-assembly}. In the absence of a mathematically rigorous definition, self-assembly is colloquially thought of as the process through which simple, unorganized components spontaneously combine, according to local interaction rules, to form some kind of organized final structure. A major objective of nanotechnology is to harness the power of self-assembly, perhaps for the purpose of engineering atomically precise medical, digital and mechanical components at the nanoscale. One strategy for doing so, developed by Nadrian Seeman, is \emph{DNA tile self-assembly} \cite{Seem82, Seem90}.

In DNA tile self-assembly, the fundamental components are ``tiles'', which are comprised of interconnected DNA strands. Remarkably, these DNA tiles can be ``programmed'', via the careful configuration of their constituent DNA strands, to automatically coalesce into a desired target structure, the characteristics of which are completely determined by the ``programming'' of the DNA tiles. In order to fully realize the power of DNA tile self-assembly, we must study the algorithmic and mathematical underpinnings of tile self-assembly.

Perhaps the simplest mathematical model of algorithmic tile self-assembly is Erik Winfree's abstract Tile Assembly Model (aTAM) \cite{Winf98}. The aTAM is a deliberately over-simplified, combinatorial model of nanoscale (DNA) tile self-assembly that ``effectivizes'' classical Wang tiling \cite{Wang61} in the sense that the former augments the latter with a mechanism for sequential ``growth'' of a tile assembly. Very briefly, in the aTAM, the fundamental components are un-rotatable, translatable square ``tile types'' whose sides are labeled with (alpha-numeric) glue ``colors'' and (integer) ``strengths''. Two tiles that are placed next to each other \emph{interact} if the glue colors on their abutting sides match, and they \emph{bind} if the strengths on their abutting sides match and sum to at least a certain (integer) ``temperature''. Self-assembly starts from a ``seed'' tile type, typically assumed to be placed at the origin, and proceeds nondeterministically and asynchronously as tiles bind to the seed-containing assembly one at a time.

Despite its deliberate over-simplification, the aTAM is a computationally expressive model. For example, Winfree \cite{Winf98} proved that the model is Turing universal, which means that, in principle, the process of self-assembly can be directed by any algorithm. In this paper, we will specifically study the extent to which tile sets in the aTAM can be algorithmically directed to ``strictly'' self-assemble (place tiles at and only at locations that belong to) shapes that are discrete self-similar tree fractals.

There are examples of prior results related to the self-assembly of fractals in the aTAM, in general \cite{jSADSSF, KautzShutters11, KL09}, as well as the strict self-assembly of tree fractals in the aTAM, specifically \cite{jSSADST, LutzShutters12}. In fact, a notable example of the latter is \cite{jSSADST}, Theorem 3.2 of which bounds from below the size of the smallest tile set in which an arbitrary shape $X$ strictly self-assembles by the depth of $X$'s largest finite sub-tree. Although not stated explicitly, an immediate corollary of Theorem 3.2 of \cite{jSSADST} is that no tree fractals strictly self-assemble in the aTAM.

While the strict self-assembly of tree fractals in the aTAM is well-understood (via Theorem 3.2 of \cite{jSSADST}), nothing is known about the strict self-assembly of ``scaled-up'' versions of tree fractals (``scaled-up'' meaning each point in the original shape is replaced by a $c \times c$ block of points). After all, the scaled-up version of any shape -- tree or otherwise -- is \emph{not} a tree in the sense of the ``full connectivity graph'' of the shape, i.e., each point in the shape is represented by one vertex and edges exist between vertices that represent adjacent points in the shape. This means that prior proof techniques, which exploit the intricate geometry of tree fractals (e.g., \cite{jSSADST, LutzShutters12}), simply cannot be applied to scaled-up versions of tree fractals. Thus, in this paper, we ask if it is possible for a scaled-up version of a tree fractal to strictly self-assemble in the aTAM.

The main contribution of this paper provides an answer to the previous question, perhaps not too surprisingly to readers familiar with the aTAM, in the negative: we prove that there is no tree fractal that strictly self-assembles in the aTAM -- at any positive scale factor. Thus, our main result generalizes Theorem 3.4 of \cite{jSSADST}, which says that the Sierpinski triangle, perhaps the most famous, well-studied example of a tree fractal, does not strictly self-assemble at scale factor $1$. Our proof makes crucial use of a recent technical lemma developed by Meunier, Patitz, Summers, Theyssier, Winslow and Woods \cite{WindowMovieLemma}, known as the ``Window Movie Lemma'' (WML), which gives a sufficient condition for taking any pair of tile assemblies, at any temperature, and ``splicing'' them together to create a new valid tile assembly. The WML is, in some sense, a pumping lemma for self-assembly that mitigates the need to use overly-complicated, convoluted case-analyses that typically arise when doing impossibility proofs in self-assembly.

What follows is a list of the main technical contributions presented in -- and the general outline of -- this paper:
\begin{itemize}
    \item In Section~\ref{sec:tree-fractals}, we exhibit a natural
      characterization of a tree fractal in terms of a few simple,
      easily checkable geometric properties of its generator. While
      perhaps well-known, this type of characterization, to the best
      of our knowledge, has yet to be explicitly documented or proved
      in the literature.
    \item In Section~\ref{sec:window-movie-lemma}, we develop a
      modified version of the general WML. Our version of the WML,
      which we call the ``Closed Window Movie Lemma'', allows us to
      replace one portion of a tile assembly with another, assuming
      certain extra ``containment'' conditions are met. Moreover,
      unlike in the original WML that lacks the extra containment
      assumptions, the replacement of one tile assembly with another
      in our Closed WML only goes ``one way'', i.e., the part
      of the tile assembly being used to replace another part cannot
      itself be replaced by the part of the tile assembly it is
      replacing.  
    \item In Section~\ref{sec:main-result}, we use our closed WML to prove that any scaled-up version of any tree fractal does not strictly self-assemble in the aTAM at any temperature. Our main result generalizes the claim that every tree fractal, at scale factor 1, does not strictly self-assemble in the aTAM (an implicit corollary to the main negative result of \cite{jSSADST}).
\end{itemize}

\section{Definitions}\label{sec-definitions}
In this section, we give a formal definition of Erik Winfree's abstract Tile Assembly Model (aTAM), define and characterize tree fractals and develop a ``Closed'' Window Movie Lemma.
\subsection{Formal description of the abstract Tile Assembly Model}
\label{sec:tam-formal}
This section gives a formal definition of the abstract Tile Assembly Model (aTAM)~\cite{Winf98}. For readers unfamiliar with the aTAM,~\cite{Roth01} gives an excellent introduction to the model.

Fix an alphabet $\Sigma$.
$\Sigma^*$ is the set of finite strings over $\Sigma$. Let $\Z$, $\Z^+$, and $\N$ denote the set of integers, positive integers, and nonnegative integers, respectively. Given $V \subseteq \Z^2$, the \emph{full grid graph} of $V$ is the undirected graph $\fullgridgraph_V=(V,E)$,
such that, for all $\vec{x}, \vec{y}\in V$, $\left\{\vec{x},\vec{y}\right\} \in E \iff \| \vec{x} - \vec{y}\| = 1$, i.e., if and only if $\vec{x}$ and $\vec{y}$ are adjacent in the $2$-dimensional integer Cartesian space.

A \emph{tile type} is a tuple $t \in (\Sigma^* \times \N)^{4}$, e.g., a unit square, with four sides, listed in some standardized order, and each side having a \emph{glue} $g \in \Sigma^* \times \N$ consisting of a finite string \emph{label} and a nonnegative integer \emph{strength}.

We assume a finite set of tile types, but an infinite number of copies of each tile type, each copy referred to as a \emph{tile}. A tile set is a set of tile types and is usually denoted as $T$.

A {\em configuration} is a (possibly empty) arrangement of tiles on the integer lattice $\Z^2$, i.e., a partial function $\alpha:\Z^2 \dashrightarrow T$.
Two adjacent tiles in a configuration \emph{interact}, or are \emph{attached}, if the glues on their abutting sides are equal (in both label and strength) and have positive strength.
Each configuration $\alpha$ induces a \emph{binding graph} $\bindinggraph_\alpha$, a grid graph whose vertices are positions occupied by tiles, according to $\alpha$, with an edge between two vertices if the tiles at those vertices interact. An \emph{assembly} is a connected, non-empty configuration, i.e., a partial function $\alpha:\Z^2 \dashrightarrow T$ such that $\fullgridgraph_{\dom \alpha}$ is connected and $\dom \alpha \neq \emptyset$.

Given $\tau\in\Z^+$, $\alpha$ is \emph{$\tau$-stable} if every cut-set
of~$\bindinggraph_\alpha$ has weight at least $\tau$, where the weight
of an edge is the strength of the glue it represents.\footnote{A
  \emph{cut-set} is a subset of edges in a graph which, when removed from
  the graph, produces two or more disconnected subgraphs. The
  \emph{weight} of a cut-set is the sum of the weights of all of the edges
  in the cut-set.}  When $\tau$ is clear from context, we say $\alpha$ is
\emph{stable}.  Given two assemblies $\alpha,\beta$, we say $\alpha$
is a \emph{subassembly} of $\beta$, and we write $\alpha \sqsubseteq
\beta$, if $\dom\alpha \subseteq \dom\beta$ and, for all points $p \in
\dom\alpha$, $\alpha(p) = \beta(p)$. For two non-overlapping
assemblies $\alpha$ and $\beta$, $\alpha \cup \beta$ is defined as the
unique assembly $\gamma$ satisfying, for all $\vec{x} \in
\dom{\alpha}$, $\gamma(\vec{x}) = \alpha(\vec{x})$, for all $\vec{x}
\in \dom{\beta}$, $\gamma(\vec{x}) = \beta(\vec{x})$, and
$\gamma(\vec{x})$ is undefined at any point $\vec{x} \in \Z^2
\backslash \left( \dom{\alpha} \cup \dom{\beta} \right)$.

A \emph{tile assembly system} (TAS) is a triple $\mathcal{T} = (T,\sigma,\tau)$, where $T$ is a tile set, $\sigma:\Z^2 \dashrightarrow T$ is the finite, $\tau$-stable, \emph{seed assembly}, and $\tau\in\Z^+$ is the \emph{temperature}.

Given two $\tau$-stable assemblies $\alpha,\beta$, we write $\alpha \to_1^{\mathcal{T}} \beta$ if $\alpha \sqsubseteq \beta$ and $|\dom\beta \setminus \dom\alpha| = 1$. In this case we say $\alpha$ \emph{$\mathcal{T}$-produces $\beta$ in one step}. If $\alpha \to_1^{\mathcal{T}} \beta$, $ \dom\beta \setminus \dom\alpha=\{p\}$, and $t=\beta(p)$, we write $\beta = \alpha + (p \mapsto t)$.
The \emph{$\mathcal{T}$-frontier} of $\alpha$ is the set $\partial^\mathcal{T} \alpha = \bigcup_{\alpha \to_1^\mathcal{T} \beta} (\dom\beta \setminus \dom\alpha$), the set of empty locations at which a tile could stably attach to $\alpha$. The \emph{$t$-frontier} $\partial^\mathcal{T}_t \alpha \subseteq \partial^\mathcal{T} \alpha$ of $\alpha$ is the set $\setr{p\in\partial^\mathcal{T} \alpha}{\alpha \to_1^\mathcal{T} \beta \text{ and } \beta(p)=t}.$

Let $\mathcal{A}^T$ denote the set of all assemblies of tiles from $T$, and let $\mathcal{A}^T_{< \infty}$ denote the set of finite assemblies of tiles from $T$.
A sequence of $k\in\Z^+ \cup \{\infty\}$ assemblies $\alpha_0,\alpha_1,\ldots$ over $\mathcal{A}^T$ is a \emph{$\mathcal{T}$-assembly sequence} if, for all $1 \leq i < k$, $\alpha_{i-1} \to_1^\mathcal{T} \alpha_{i}$.
The {\em result} of an assembly sequence $\vec{\alpha}$, denoted as $\textmd{res}(\vec{\alpha})$, is the unique limiting assembly (for a finite sequence, this is the final assembly in the sequence).

We write $\alpha \to^\mathcal{T} \beta$, and we say $\alpha$ \emph{$\mathcal{T}$-produces} $\beta$ (in 0 or more steps) if there is a $\mathcal{T}$-assembly sequence $\alpha_0,\alpha_1,\ldots$ of length $k = |\dom\beta \setminus \dom\alpha| + 1$ such that
(1) $\alpha = \alpha_0$,
(2) $\dom\beta = \bigcup_{0 \leq i < k} \dom{\alpha_i}$, and
(3) for all $0 \leq i < k$, $\alpha_{i} \sqsubseteq \beta$.
If $k$ is finite then it is routine to verify that $\beta = \alpha_{k-1}$.

We say $\alpha$ is \emph{$\mathcal{T}$-producible} if $\sigma \to^\mathcal{T} \alpha$, and we write $\prodasm{\mathcal{T}}$ to denote the set of $\mathcal{T}$-producible assemblies. The relation $\to^\mathcal{T}$ is a partial order on $\prodasm{\mathcal{T}}$ \cite{Roth01,jSSADST}.

An assembly $\alpha$ is \emph{$\mathcal{T}$-terminal} if $\alpha$ is $\tau$-stable and $\partial^\mathcal{T} \alpha=\emptyset$.
We write $\termasm{\mathcal{T}} \subseteq \prodasm{\mathcal{T}}$ to denote the set of $\mathcal{T}$-producible, $\mathcal{T}$-terminal assemblies. If $|\termasm{\mathcal{T}}| = 1$ then  $\mathcal{T}$ is said to be {\em directed}.

We say that a TAS $\mathcal{T}$ \emph{strictly (a.k.a. uniquely) self-assembles} $X \subseteq \Z^2$ if, for all $\alpha \in \termasm{\mathcal{T}}$, $\dom\alpha = X$; i.e., if every terminal assembly produced by $\mathcal{T}$ places tiles on -- and only on -- points in the set $X$.

In this paper, we consider scaled-up versions of subsets of $\Z^2$. Formally, if $X$ is a subset of $\Z^2$ and $c \in \Z^+$, then a $c$-\emph{scaling} of $X$ is defined as the set $X^c = \left\{ (x,y) \in \mathbb{Z}^2 \; \left| \; \left( \left\lfloor \frac{x}{c} \right\rfloor, \left\lfloor \frac{y}{c} \right\rfloor \right) \in X \right.\right\}$. Intuitively, $X^c$ is the subset of $\Z^2$  obtained by replacing each point in $X$ with a $c \times c$ block of points. We refer to the natural number $c$ as the \emph{scaling factor} or \emph{resolution loss}.

\subsection{Discrete self-similar tree fractals}\label{sec:tree-fractals}
In this section, we introduce a new formal characterization of
discrete self-similar tree fractals. The proof of
Theorem~\ref{thm:tree-fractal} below is omitted from this version of the
paper due to lack of space.

\begin{notation}
We use $\mathbb{N}_g$ to denote the subset $\{0, \ldots, g-1\}$ of
$\mathbb{N}$.
\end{notation}

\begin{notation}
If $A$ and $B$ are subsets of $\N^2$ and $k\in \N$, then $A+kB = \{\vec{m}+k\vec{n}~|~\vec{m}\in A$ and $\vec{n}\in B\}$.
\end{notation}

The following definition is adapted from \cite{jSADSSF}.

\begin{definition}\label{def:dssf}
Let $1<g \in \mathbb{N}$ and $\mathbf{X} \subset \mathbb{N}^2$. We say that
$\mathbf{X}$ is a \emph{$g$-discrete self-similar fractal} (or \emph{$g$-dssf}
for short), if there is a set $\{(0,0)\} \subset G \subset
\mathbb{N}_g^2$ with at least one point in every row and column, such
that $\displaystyle \mathbf{X} = \bigcup_{i=1}^{\infty}X_i$, where $X_i$, the
$i^{th}$ \emph{stage} of $\mathbf{X}$, is defined by $X_1= G$ and
$X_{i+1}= X_i\ +\ g^iG$. We say that $G$ is the
\emph{generator} of $\mathbf{X}$.
\end{definition}

Intuitively, a $g$-dssf is built as follows. Start by selecting points
in $\N_g^2$ satisfying the constraints listed in
Definition~\ref{def:dssf}. This first stage of the fractal is the
generator. Then, each subsequent stage of the fractal is obtained by
adding a full copy of the previous stage for every point in the
generator and translating these copies so that their relative
positions are identical to the relative positions of the individual
points in the gnerator.

In this paper, we focus on \emph{tree} fractals, that is, fractals whose
underlying graph is a tree. We introduce terminology and notation that
will help us in formulating a complete characterization of tree
fractals in terms of geometric properties of their generator.

\begin{definition}
Let $S$ be any finite subset of $\mathbb{Z}^2$. Let $l$, $r$, $b$, and
$t$ denote the following integers:

\centerline{$\displaystyle l_S = \min_{(x,y) \in S} x\qquad r_S =
  \max_{(x,y) \in S} x\qquad b_S = \min_{(x,y) \in S} y\qquad t_S =
  \max_{(x,y) \in S} y$}

An \emph{h-bridge} of $S$ is any subset of $S$ of the form
$hb_S(y) = \{(l_S,y),(r_S,y)\}$. Similarly, a \emph{v-bridge} of $S$ is any
subset of $S$ of the form $vb_S(x) = \{(x,b_S),(x,t_S)\}$. We say that a bridge is \emph{connected} if there is a simple path in $S$ connecting the two bridge points.
\end{definition}

\begin{notation}
Let $S$ be any finite subset of $\Z^2$. We will denote by
$nhb_S$ and $nvb_S$, respectively, the number of h-bridges and
the number of v-bridges of $S$.
\end{notation}

The following theorem is a new characterization of tree fractals in
terms of simple connectivity properties of their generator.

\both{
\begin{theorem}\label{thm:tree-fractal}
$\displaystyle\mathbf{T} = \bigcup_{i=1}^{\infty}{T_i}$ is a $g$-discrete self-similar tree fractal, for some $g > 1$, with generator $G$ if and only if \\
\hspace*{1cm}a. $G$ is a tree, and\\
\hspace*{1cm}b. $nhb_G = nvb_G =1$
\end{theorem}
}

\begin{notation}
The directions $\mathcal{D} = \{N,E,S,W\}$ will be used as functions
from $\mathbb{Z}^2$ to $\mathbb{Z}^2$ such that $N(x,y) = (x,y+1)$,
$E(x,y) = (x+1,y)$, $S(x,y) = (x,y-1)$ and $W(x,y) = (x-1,y)$. Note
that $N^{-1} = S$ and $W^{-1}=E$.
\end{notation}

\begin{notation}
Let $X \subseteq \mathbb{Z}^2$. We say that a point $(x,y) \in X$ is
$D$-\emph{free} in $X$, for some direction $D$, if $D(x,y) \not \in
X$.
\end{notation}

\begin{definition}
Let $G$ be the generator of any $g$-discrete self-similar fractal.  A
\emph{pier} is a point in $G$ that is $D$-free for exactly three of
the four directions in $\mathcal{D}$. We say that a pier $(p,q)$ is
\emph{$D$-pointing} if $D^{-1}(p,q) \in G$. Note that a pier always points
in exactly one direction
\end{definition}

Finally, the following observation follows from the
fact that a tree with more than one vertex must contain at least two
leaf nodes.

\begin{observation}
\label{obs:piers}
If $G$ is the generator of any discrete self-similar fractal and $G$ is a tree, then it must
contain at least two piers.
\end{observation}

\subsection{The Closed Window Movie Lemma}\label{sec:window-movie-lemma}
In this subsection, we develop a more accommodating (modified) version
of the general Window Movie Lemma (WML) \cite{WindowMovieLemma}. Our
version of the WML, which we call the ``Closed Window Movie Lemma'',
allows us to replace one portion of a tile assembly with another,
assuming certain extra ``containment'' conditions are met. Moreover,
unlike in the original WML that lacks the extra containment
assumptions, the replacement of one tile assembly with another in our
Closed WML only goes ``one way'', i.e., the part of the tile assembly
being used to replace another part cannot itself be replaced by the
part of the tile assembly it is replacing. We must first define some
notation that we will use in our closed Window Movie Lemma.

A window $w$ is a set of edges forming a cut-set of the full grid
graph of $\mathbb{Z}^2$. For the purposes of this paper, we say that a
\emph{closed window} $w$ induces a cut\footnote{A \emph{cut} is
  a partition of the vertices of a graph into two disjoint subsets
  that are joined by at least one edge.}
 of the full grid graph of
$\mathbb{Z}^2$, written as $C_w = (C_{<\infty},C_\infty)$, where
$C_{\infty}$ is infinite, $C_{<\infty}$ is finite and for all pairs of
points $\vec{x},\vec{y} \in C_{<\infty}$, every simple path connecting
$\vec{x}$ and $\vec{y}$ in the full grid graph of $C_{<\infty}$ does
not cross the cut $C_w$. We call the set of vertices that make up
$C_{<\infty}$ the \emph{inside} of the window $w$, and write
$inside(w) = C_{<\infty}$ and $outside(w) = \mathbb{Z}^2 \backslash\,inside(w) =
C_\infty$. We say that a window $w$ is \emph{enclosed} in another
window $w'$ if $inside(w) \subseteq inside(w')$.

Given a window $w$ and an assembly $\alpha$, a window that {\em
  intersects} $\alpha$ is a partitioning of $\alpha$ into two
configurations (i.e., after being split into two parts, each part may
or may not be disconnected). In this case we say that the window $w$
cuts the assembly $\alpha$ into two configurations $\alpha_L$
and~$\alpha_R$, where $\alpha = \alpha_L \cup \alpha_R$. For
notational convenience, if $w$ is a closed window, we write $\alpha_I$
for the assembly inside $w$ and $\alpha_O$ for the assembly outside
$w$.  Given a window $w$, its translation by a vector $\vec{c}$,
written $ w + \vec{c}$ is simply the translation of each of $w$'s
elements (edges) by~$\vec{c}$.


For a window $w$ and an assembly sequence $\vec{\alpha}$, we define a window movie~$M$ to be the order of placement, position and glue type for each glue that appears along the window $w$ in $\vec{\alpha}$. Given an assembly sequence $\vec{\alpha}$ and a window $w$, the associated {\em window movie} is the maximal sequence $M_{\vec{\alpha},w} = (v_{0}, g_{0}) , (v_{1}, g_{1}), (v_{2}, g_{2}), \ldots$ of pairs of grid graph vertices $v_i$ and glues $g_i$, given by the order of the appearance of the glues along window $w$ in the assembly sequence $\vec{\alpha}$.
Furthermore, if $k$ glues appear along $w$ at the same instant (this happens upon placement of a tile that has multiple  sides  touching $w$) then these $k$ glues appear contiguously and are listed in lexicographical order of the unit vectors describing their orientation in $M_{\vec{\alpha},w}$.

Let $w$ be a window and $\vec{\alpha}$ be an assembly sequence and $M = M_{\vec{\alpha},w}$. We use the notation $\mathcal{B}\left(M\right)$ to denote the \emph{bond-forming submovie} of $M$, i.e., a restricted form of $M$ that consists of only those steps of $M$ that place glues that eventually form positive-strength bonds in the assembly $\alpha = \res{\vec{\alpha}}$. Note that every window movie has a unique bond-forming submovie.

\begin{lemma}[Closed Window Movie Lemma]\label{lem:wml}
Let $\vec{\alpha}~=~(\alpha_i~|~0\leq~i<~l)$, with $l \in \mathbb{Z}^+
\cup \{\infty\}$, be an assembly sequence in some TAS $\mathcal{T}$
with result $\alpha$. Let $w$ be a closed window that partitions
$\alpha$ into $\alpha_I$ and $\alpha_O$, and $w'$ be a closed window
that partitions $\alpha$ into $\alpha_I'$ and $\alpha_O'$.  If
$\mathcal{B}(M_{\vec{\alpha},w}) + \vec{c} =
\mathcal{B}(M_{\vec{\alpha},w'})$ for some $\vec{c}\neq(0,0)$ and
the window $w+\vec{c}$ is enclosed in $w'$, then the assembly $\alpha'_O \cup (\alpha_I + \vec{c})$ is
in $\mathcal{A[\mathcal{T}]}$.
\end{lemma}

\begin{proof}
Before we proceed with the proof, the next paragraph introduces some
notation taken directly from \cite{WindowMovieLemma}.

For an assembly sequence $\vec{\alpha} = (\alpha_i \mid 0 \leq i < l)$, we write $\left| \vec{\alpha} \right| = l$ (note that if $\vec{\alpha}$ is infinite, then $l = \infty$). We write $\vec{\alpha}[i]$ to denote $\vec{x} \mapsto t$, where $\vec{x}$ and~$t$ are such that $\alpha_{i+1} = \alpha_i + \left(\vec{x} \mapsto t\right)$, i.e., $\vec{\alpha}[i]$ is the placement  of tile type $t$ at position~$\vec{x}$, assuming that $\vec{x} \in \partial_{t}\alpha_i$. We write $\vec{\alpha}[i] + \vec{c}$, for some vector $\vec{c}$, to denote $\left(\vec{x}+\vec{c}\right) \mapsto t$. We define $\vec{\alpha} = \vec{\alpha} + \left(\vec{x} \mapsto t\right) = (\alpha_i \mid 0 \leq i < k + 1)$, where $\alpha_{k} = \alpha_{k-1} + \left(\vec{x} \mapsto t\right)$ if $\vec{x} \in \partial_{t}\alpha_{k-1}$ and undefined otherwise, assuming $\left| \vec{\alpha} \right| > 0$. Otherwise, if $\left| \vec{\alpha} \right| = 0$, then $\vec{\alpha} = \vec{\alpha} + \left(\vec{x} \mapsto t \right) = (\alpha_0)$, where $\alpha_0$ is the assembly such that $\alpha_0\left(\vec{x}\right) = t$ and is undefined at all other positions. This is our notation for appending steps to the assembly sequence $\vec{\alpha}$: to do so, we must specify a tile type $t$ to be placed at a given location $\vec{x} \in \partial_t\alpha_{i}$. If $\alpha_{i+1} = \alpha_i + \left(\vec{x} \mapsto t\right)$, then we write $Pos\left(\vec{\alpha}[i]\right) = \vec{x}$ and $Tile\left(\vec{\alpha}[i]\right) = t$. For a window movie $M = (v_0,g_0), (v_1,g_1), \ldots$, we write $M[k]$ to be the pair $\left(v_{k},g_{k}\right)$ in the enumeration of $M$ and $Pos\left(M[k]\right) = v_{k}$, where $v_{k}$ is a vertex of a grid graph.

We now proceed with the proof, throughout which we will assume that $M = \mathcal{B}\left(M_{\vec{\alpha},w}\right)$ and $M' = \mathcal{B}\left(M_{\vec{\alpha},w'}\right)$. Since $M + \vec{c} =
M'$ for some $\vec{c}\neq(0,0)$ and $w$ and $w'$ are both closed windows, it must be the case that the seed tile of $\alpha$ is in $\dom{\alpha_O} \cap \dom{\alpha'_O}$ or in $\dom{\alpha_I} \cap \dom{\alpha'_I}$. In other words, the seed tile cannot be in $\dom{\alpha_I} \backslash\, \dom{\alpha'_I}$ nor in $\dom{\alpha'_I} \backslash\, \dom{\alpha_I}$. Therefore, assume without loss of generality that the seed tile is in $\dom{\alpha_O} \cap \dom{\alpha'_O}$.

The algorithm in Figure \ref{fig:algo-seq} describes how to produce a new valid assembly sequence $\vec{\gamma}$.

%
%
%
%
%
%
%
%
%
%

\begin{figure}[htp]
\begin{algorithm}[H]
\SetAlgoLined
Initialize $i$, $j= 0$ and $\vec{\gamma}$ to be empty

\For{$k = 0$ \KwTo $|M| - 1$}{
  \If{$Pos(M'[k]) \in \dom{\alpha'_O}$}{
    \While{$Pos(\vec{\alpha}[i])\neq Pos(M'[k])$}{
      \If{$Pos(\vec{\alpha}[i]) \in \dom{\alpha'_O}$}{$\vec{\gamma} = \vec{\gamma} + \vec{\alpha}[i]$}
      $i = i + 1$
    }
    $\vec{\gamma} = \vec{\gamma} + \vec{\alpha}[i]$

    $i = i + 1$
  }
  \Else {
    \While{$Pos(\vec{\alpha}[j])\neq Pos(M[k])$}{
      \If{$Pos(\vec{\alpha}[j]) \in \dom{\alpha_I}$}{
        $\vec{\gamma} = \vec{\gamma} + \left(\vec{\alpha}[j] + \vec{c}\right)$}

        $j = j + 1$
    }
    $\vec{\gamma} = \vec{\gamma} + \vec{\alpha}[j]$

    $j = j + 1$

  }

}

\While{$inside(w) \cap \partial \res{\vec{\gamma}} \ne \emptyset$}{
    \If{$Pos(\vec{\alpha}[j]) \in \dom{\alpha_I}$}{$\vec{\gamma} = \vec{\gamma} + (\vec{\alpha}[j] + \vec{c})$}

    $j = j + 1$
    }

\While {$i < |\vec{\alpha}|$}{
  \If{$Pos(\vec{\alpha}[i]) \in \dom{\alpha'_O}$}{$\vec{\gamma} = \vec{\gamma} + \vec{\alpha}[i]$}

    $i = i + 1$
}

\Return $\vec{\gamma}$

\end{algorithm}
\caption{The algorithm to produce a valid assembly sequence $\vec{\gamma}$.}
\label{fig:algo-seq}
\vspace{-20pt}
\end{figure}

If we assume that  the assembly sequence $\vec{\gamma}$ ultimately produced by the algorithm is valid, then the result of $\vec{\gamma}$ is indeed $\alpha'_O \cup \left( \alpha_I + \vec{c} \right)$. Observe that $\alpha_I$ must be finite, which implies that $M$ is finite. If $|\vec{\alpha}| < \infty$, then all loops will terminate. If $|\vec{\alpha}| = \infty$, then $|\alpha'_O| = \infty$ and the first two loops will terminate and the last loop will run forever. In either case, for every tile in~$\alpha'_O$ and $\alpha_I + \vec{c}$, the algorithm adds a step to the sequence $\vec{\gamma}$ involving the addition of this tile to the assembly. However, we need to prove that the assembly sequence $\vec{\gamma}$ is valid. It may be the case that either: 1. there is insufficient bond strength between the tile to be placed and the existing neighboring tiles, or 2. a tile is already present at this location.

\textbf{Case 1:}
In this case, we claim the following: at each step of the algorithm, the current version of $\vec{\gamma}$ is a valid assembly sequence whose result is a producible subassembly of $\alpha'_O \cup \left(\alpha_I + \vec{c}\right)$. Note that three loops in the algorithm iterate through all steps of $\vec{\alpha}$, such that at any time when adding $\vec{\alpha}[i]$ (or $\vec{\alpha}[j] + \vec{c}$) to $\vec{\gamma}$, all steps of the window movie occurring before $\vec{\alpha}[i]$ (or $\vec{\alpha}[j]$) in $\vec{\alpha}$ have occurred. Similarly, all tiles in $\alpha'_O$ (or $\alpha_I + \vec{c}$) added to $\alpha$ before step $i$ in the assembly sequence have occurred.

So, if the tile $Tile\left(\vec{\alpha}[i]\right)$ that is added to the subassembly of $\alpha$ produced after $i-1$ steps can bond at a location in $\alpha'_O$ to form a $\tau$-stable assembly, then the same tile added to the producible assembly of  $\vec{\gamma}$ must also bond to the same location in $\vec{\gamma}$, as the neighboring glues consist of (i) an identical set of glues from tiles in the subassembly of $\alpha'_O$ and (ii) glues on the side of the window movie containing~$\alpha_I + \vec{c}$.  Similarly, the tiles of $\alpha_I + \vec{c}$ must also be able to bind.

\textbf{Case 2:} Since we only assume that $\mathcal{B}\left(M_{\vec{\alpha},w}\right) + \vec{c} = \mathcal{B}\left(M_{\vec{\alpha},w'}\right)$, as opposed to the stronger condition $\mathcal{B}\left(M_{\vec{\alpha},w+\vec{c}}\right) = \mathcal{B}\left(M_{\vec{\alpha},w'}\right)$, which is assumed in the original WML, we must show that $\dom{\left(\alpha_I + \vec{c}\right)} \cap \dom{\alpha'_O} = \emptyset$. To see this, observe that, by assumption, $w + \vec{c}$ is enclosed in $w'$, which, by definition, means that $inside\left(w+\vec{c}\right) \subseteq inside(w')$. Then we have $\vec{x} \in \dom{\alpha'_O} \Rightarrow \vec{x} \in outside(w') \Rightarrow \vec{x} \not \in inside\left(w'\right) \Rightarrow \vec{x} \not \in inside\left(w+\vec{c}\right) \Rightarrow \vec{x} \not \in \dom{\left(\alpha_I + \vec{c}\right)}$. Thus, locations in $\alpha_I + \vec{c}$ only have tiles from $\alpha_I$ placed in them, and locations in $\alpha'_O$ only have tiles from $\alpha'_O$ placed in them.

So the assembly sequence of $\vec{\gamma}$ is valid, i.e., every addition to $\vec{\gamma}$ adds a tile to the assembly to form a new producible assembly. Since we have a valid assembly sequence, as argued above, the finished producible assembly is~$\alpha'_O \cup \left(\alpha_I + \vec{c}\right)$.
\qed
\end{proof}

\section{Main result: scaled tree fractals do not strictly self-assemble}\label{sec:main-result}

In this section, we first define some notation and then prove our main
result.

\subsection{Notation}

Recall that each stage $X_s$ ($s>1$) of a $g$-dssf (scaled by a factor
$c$) is made up of copies of the previous stage $X_{s-1}$, each of
which is a square of size $cg^{s-1}$.
In the proof of our main result, we will need to refer to one of
the squares of size $cg^{s-2}$ inside the copies of stage $X_{s-1}$, leading to
the following notation.

\begin{notation}\label{smallwindows}
Let $c \in \Z^+$, $1<s \in \N$ and $1<g \in
\N$. Let $e, f, p, q \in \N_g$. We use
$S_s^c(e,f,p,q)$ to denote $\{0,1,\ldots,cg^{s-2}-1\}^2+ cg^{s-1}(e,f)
+ cg^{s-2}(p,q)$ and
$W_s^c(e,f,p,q)$ to denote the square-shaped, closed window whose inside
is $S_s^c(e,f,p,q)$.
\end{notation}

In Figure~\ref{fig:case2example} below, the small and large red
windows are $W_2^1(0,1,3,2)$ and $W_3^1(0,1,3,2)$, respectively.

Next, we will need to translate a small window to a position inside a
larger window. These two windows will correspond to squares at the
same relative position in different stages $i$ and $j$ of a $g$-dssf.

\begin{notation}\label{translations}
Let $c \in \mathbb{Z}^+$, $i,j \in \N\,\backslash\{0,1\}$, with $i<j$, and $e, f,p, q \in
\N_g$.
We use $\vec{t}^c_{i\rightarrow j}(e,f,p,q)$ to denote
the vector joining the southwest corner of $W_i^c(e,f,p,q)$ to the
southwest corner of $W_j^c(e,f,p,q)$. In other words, $\vec{t}^c_{i\rightarrow
  j}(e,f,p,q) = \left(c\left(g^{j-1}-g^{i-1}\right)e+c\left(g^{j-2}-g^{i-2}\right)p,
c\left(g^{j-1}-g^{i-1}\right)f+c\left(g^{j-2}-g^{i-2}\right)q\right)$.
\end{notation}

For example,
in Figure~\ref{fig:case2example} below,
$\vec{t}_{2\rightarrow 3}^1(0,1,3,2) = (9,18)$.

Finally, to apply Lemma~\ref{lem:wml}, we will need the bond-forming
submovies to line up. Therefore, once the two square windows share
their southwest corner after using the translation defined above, we
will need to further translate the smallest one either up or to the
right, or both, depending on which side of the windows contains the
bond-forming glues, which, in the case of scaled tree fractals, always
form a straight (vertical or horizontal) line of length $c$. We will
compute the coordinates of this second translation in our main
proof. For now, we establish an upper bound on these coordinates that
will ensure that the translated window will remain enclosed in the
larger window.

\both{
\begin{lemma}\label{lem:enclosure2}
Let $c \in \mathbb{Z}^+$, $i,j \in\N\,\backslash\{0,1\}$, with $i<j$,
$e, f,p,q \in \mathbb{N}_g$, and $x,y \in \mathbb{N}$. Let $m =
c(g^{j-2}-g^{i-2})$. If $x\leq m$ and $y\leq m$, then the window
$W^c_i(e,f,p,q) +\vec{t}^c_{i\rightarrow j}(e,f,p,q)+(x,y)$ is
enclosed in the window $W^c_j(e,f,p,q)$.
\end{lemma}
}

Finally, the following lemma establishes that any scaled tree fractal
$\mathbf{T}^c $ contains an infinite number of closed windows that all
cut the fractal along a single line of glues.

\both{
\begin{lemma} \label{lem:piers}
Let $\mathbf{T}$ be any tree fractal with generator $G$. If $c \in
\mathbb{Z}^+$, then it is always possible to pick one pier $(p,q)$
and one point $(e,f)$, both in $G$, such that, for $1<s \in \N$, $W^c_s(e,f,p,q)$ encloses a configuration that is connected
to $\mathbf{T}^c$ via a single connected (horizontal or vertical) line
of glues of length $c$.
\end{lemma}
}


The proofs of the lemmas in this sub-section are omitted from this version of the paper due to lack of space.

\subsection{Application to scaled tree fractals}

The main contribution of this paper is the
following theorem.

\begin{theorem}
Let $\mathbf{T}$ be any tree fractal. If $c \in \mathbb{Z}^+$, then
$\mathbf{T}^c$ does not strictly self-assemble in the aTAM.
\end{theorem}

\begin{proof}

Let $\mathbf{T}$ be any tree fractal with a $g\times g$ generator $G$,
where $1<g \in \N$. Let $c$ be any positive integer.
For the sake of obtaining a contradiction, assume that $\mathbf{T}^c$
does strictly self-assemble in some TAS $\mathcal{T} =
(T,\sigma,\tau)$. Further assume that $\vec{\alpha}$ is some assembly
sequence in $\mathcal{T}$ whose result is $\alpha$, such that $\dom
\alpha = \mathbf{T}^c$.

\begin{figure}[h]
\centering
 \includegraphics[width=0.9\textwidth]{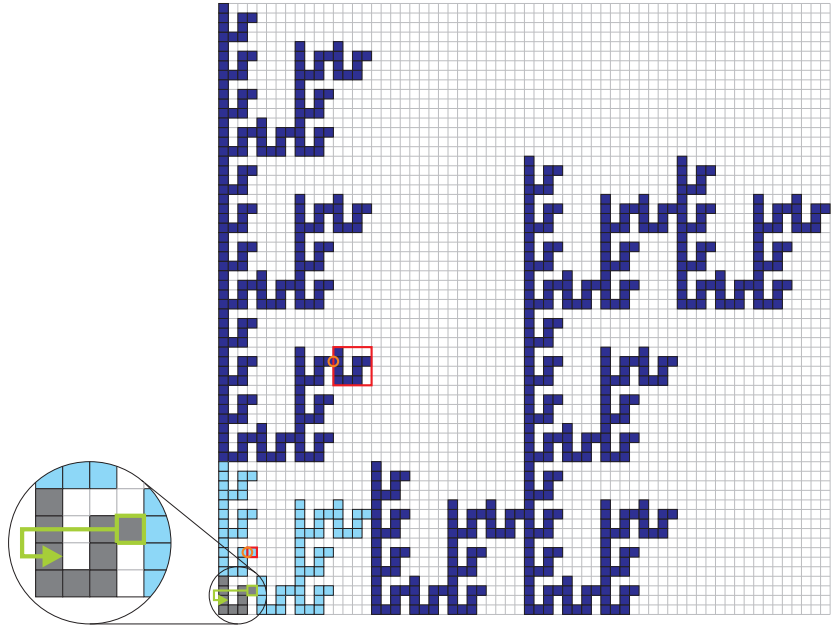}
\caption{First three stages ($s=1,2,3$) of an unscaled ($c=1$) 4-dssf
  tree fractal with an east-pointing pier at position $(3,2)$ (the
  green square). The $E$-free point $(0,1)$ is at the tip of the green
  arrow. In other words, $(p,q)=(3,2)$, and $(e,f)=(0,1)$.}
\label{fig:case2example}
\vspace{-10pt}
\end{figure}

\begin{figure}[h]
\centering
 \includegraphics[width=0.95\textwidth]{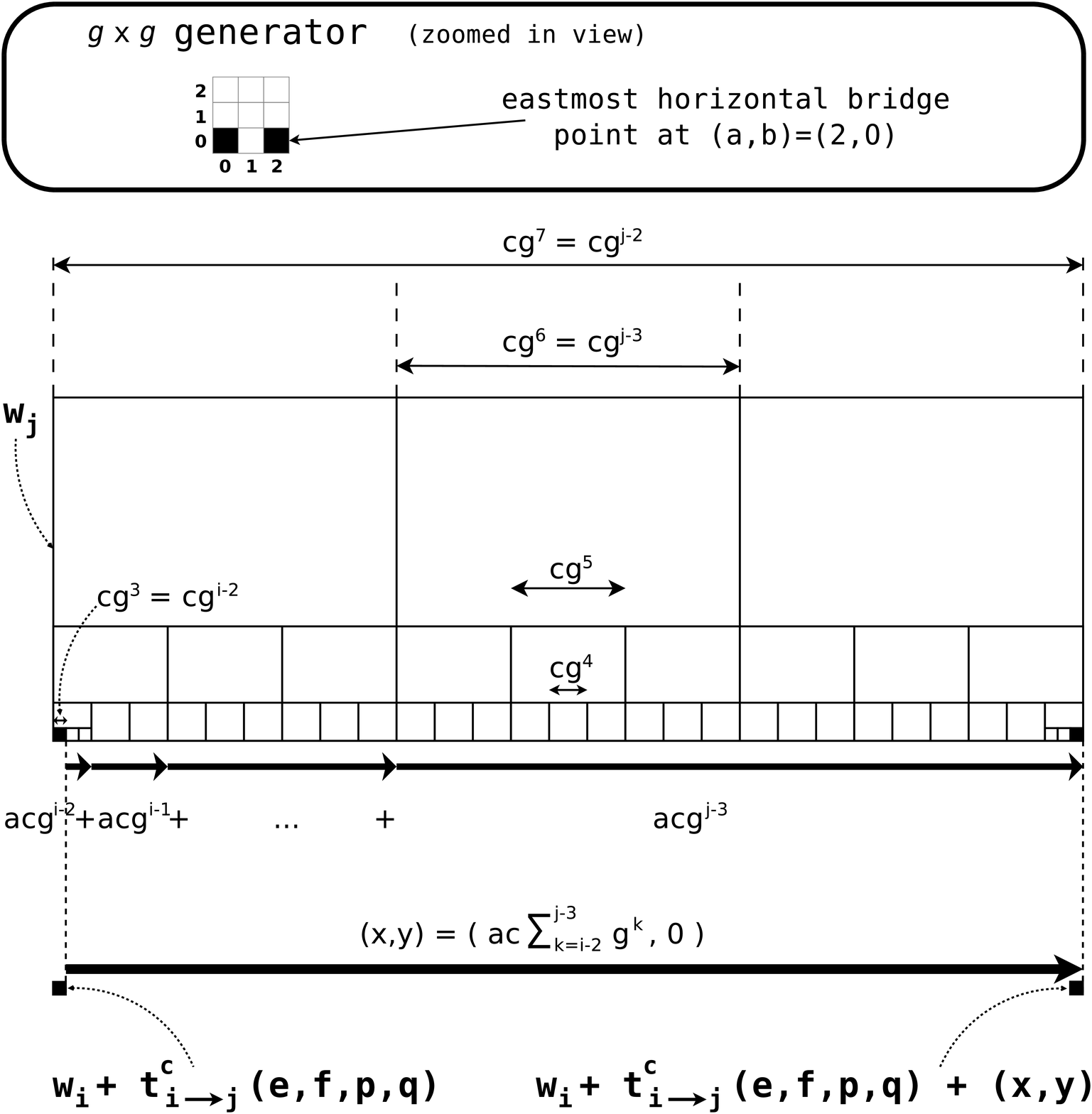}
\caption{$(x,y)$ translation needed to align $w_i$ and $w_j$ on their
  east side once their southwest corners already match. Example with a
  west-pointing pier (not shown) and $g=3$, $i=5$, $j=9$,
  $(a,b)=(2,0)$.}
\label{fig:translation}
\vspace{-10pt}
\end{figure}

According to Lemma~\ref{lem:piers}, we can always pick one pier
$(p,q)$ and a point $(e,f)$, both in $G$, such that, for $1<s \in \N$,
the window $W^c_s(e,f,p,q)$, which we will abbreviate $w_s$, encloses
a configuration that is connected to $\mathbf{T}^c$ via a single line
of glues of length $c$.\footnote{Without loss of generality, we will
  assume that this line of glues is positioned on the western side of
  the windows and is thus vertical (see the orange circles in
  Figure~\ref{fig:case2example}, where $s=2$ and $s=3$ for the small
  and large red windows, respectively, and $(p,q) = (3,2)$ and
  $(e,f)=(0,1)$), because the chosen pier in our example points
  east. A similar reasoning holds for piers pointing north, south or
  west.}  The maximum number of distinct combinations and orderings of
glue positionings along this line of glues is finite.\footnote{This
  number is $(T_{glue})^{2c}\cdot(2c)!$, where $T_{glue}$ is the total
  number of distinct glue types in $T$.}  By the generalized
pigeonhole principle, since
$\left|\left\{w_s~\left|~1<s\in\N\right.\right\}\right|$ is infinite,
there must be at least one bond-forming submovie such that an infinite
number of windows generate this submovie (up to translation). Let us
pick two such windows, say, $w_i$ and $w_j$ with $i<j$, such that
$\mathcal{B}(M_{\vec{\alpha},w_i})$ and
$\mathcal{B}(M_{\vec{\alpha},w_j})$ are equal (up to translation). We
must pick these windows carefully, since as stated in the proof of
Lemma~\ref{lem:wml}, the seed of $\alpha$ must be either in both
windows or in neither. This condition can always be satisfied. The
only case where the seed is in more than one window is when it is at
position $(0,0)$ and $e=f=p=q=0$, which implies that all windows
include the origin. In all other cases, none of the windows
overlap. So, if the seed belongs to one of them, say $w_k$, then we
can pick any $i$ greater than $k$ (and $j>i$). Finally, if the seed
does not belong to any windows, then any choice of $i$ and $j>i$ will
do.

We will now prove that $w_i$ and $w_j$ satisfy the two
conditions of Lemma~\ref{lem:wml}.

First, we compute $\vec{c}$ such that
$\mathcal{B}(M_{\vec{\alpha},w_i}) +\vec{c} =
\mathcal{B}(M_{\vec{\alpha},w_j})$. We know that $w_i
+\vec{t}^c_{i\rightarrow j}(e,f,p,q)$ and $w_j$ share their southwest
corner.  We need to perform one more translation to align the
bond-forming glues of $w_i$ and $w_j$.  We use $(a,b)$ to denote the
position of the western point in the horizontal bridge of the
generator. In our example (east-pointing pier), $a=0$ and $b$ is a
variable with domain $\N_g$ ($b=2$ in
Figure~\ref{fig:case2example}). To align the bond-forming glues of
$w_i$ and $w_j$, we must translate $w_i +\vec{t}^c_{i\rightarrow
  j}(e,f,p,q)$ by $(x,y) =
\left(0,bc\sum_{k=i-2}^{j-3}g^k\right)$. The general computation for
this translation is illustrated in Figure~\ref{fig:translation}. Since
$x\leq m$ (as defined in Lemma~\ref{lem:enclosure2}) and
$bc\sum_{k=i-2}^{j-3}g^k \leq (g-1)c\sum_{k=i-2}^{j-3}g^k =
c\left(\sum_{k=i-1}^{j-2}g^k-\sum_{k=i-2}^{j-3}g^k\right)
=c\left(g^{j-2}-g^{i-2}\right)=m$, we can apply
Lemma~\ref{lem:enclosure2} to infer that, with $\vec{c} =
\vec{t}^c_{i\rightarrow j}(e,f,p,q) + (x,y)$, $w_i+\vec{c}$ is
enclosed in $w_j$. Therefore, the second condition of
Lemma~\ref{lem:wml} holds.

Second, by construction, $\mathcal{B}(M_{\vec{\alpha},w_i}) +\vec{c} =
\mathcal{B}(M_{\vec{\alpha},w_j})$.  Therefore, the first condition of
Lemma~\ref{lem:wml} holds.

In conclusion, the two
conditions of Lemma~\ref{lem:wml} are satisfied, with $\alpha_I$ and
$\alpha_O'$ defined as the intersection of $\mathbf{T}^c$ with the inside
of $w_i$ and the outside of $w_j$, respectively. We can thus conclude
that the assembly $\alpha_I \cup (\alpha_O' -\vec{c})$ is producible
in $\mathcal{T}$. Note that this assembly is identical (up to
translation) to $\mathbf{T}^c$, except that the interior of the large
window $w_j$ is replaced by the interior of the small window
$w_i$. Since the configurations in these two windows cannot be
identical, we have proved that $\mathcal{T}$ does not strictly
self-assemble $\mathbf{T}^c$, which is a contradiction. \qed
\end{proof}

\section{Conclusion}\label{sec:conclusion}

In this paper, we made three contributions. First, we gave a new characterization of tree fractals in terms of simple geometric properties of
their generator. Second, we proved a new variant of the Window Movie
Lemma in \cite{WindowMovieLemma}, which we call the ``Closed Window
Movie Lemma.'' Third, we proved that no scaled-up version of any
discrete self-similar tree fractal strictly self-assembles in the
aTAM. In future work, we plan to extend this result to larger classes of
non-tree fractals similar to the class of pinch-point
fractals in~\cite{jSADSSF}.

\bibliographystyle{amsplain}
\bibliography{tam}

\end{document}